\documentclass[journal,12pt,onecolumn,draftclsnofoot,]{IEEEtran}
\IEEEoverridecommandlockouts
\usepackage{epsf}
\usepackage[pdftex]{graphicx}
\usepackage[outdir=./]{epstopdf}
\usepackage{amsmath,bm}
\usepackage{amssymb}
\usepackage{epsfig,latexsym,amsmath,epsf,amssymb,amsfonts}
\usepackage{verbatim}

\usepackage[english]{babel}
\usepackage[utf8]{inputenc}
\usepackage{algorithm}
\usepackage[noend]{algpseudocode}

\usepackage{placeins}

\usepackage{epstopdf}
\usepackage{amsthm}
\usepackage{url}
\usepackage[noadjust]{cite}
\usepackage{caption}
\usepackage{subcaption}
\newtheorem{theorem}{Theorem}

\usepackage{authblk}
\usepackage{amscd,mathrsfs}
\usepackage[numbers,sort]{natbib}
\usepackage[withbib,all]{authorindex}
\usepackage{esint}
\usepackage{paralist}
\usepackage{color}
\usepackage{bm}
\usepackage{amsthm}
\usepackage{mathtools}
\usepackage{mathptmx} 
\usepackage{times} 
\usepackage{amssymb}
\usepackage{float}
\usepackage{dsfont}

\begin{document}

\pagestyle{empty}

\title{Finite Horizon Throughput Maximization and Sensing Optimization in Wireless Powered Devices over Fading Channels}
\author{Mehdi Salehi Heydar Abad, Ozgur Ercetin %
\thanks{Mehdi~Salehi~Heydar~Abad and Ozgur~Ercetin are with the Faculty of Engineering and Natural Sciences, Sabanci University, Istanbul, Turkey. Email: mehdis@sabanciuniv.edu, oercetin@sabanciuniv.edu}
\thanks{This work was in part supported by EC H2020-MSCA-RISE-2015 programme under grant number 690893.}
}

\maketitle

\newtheorem{lemma}{Lemma}
\newtheorem{corollary}{Corollary}
\thispagestyle{empty}

\begin{abstract}
   Wireless power transfer (WPT) is  a promising technology that provides the network a way to replenish the batteries of the remote devices by utilizing RF transmissions. We study a class of harvest-first-transmit-later type of WPT policy, where an access point (AP) first employs RF power transfer to recharge a wireless powered device (WPD) for a certain period subjected to optimization, and then, the harvested energy is subsequently used by the WPD to transmit its data bits back to the AP over a finite horizon. A significant challenge regarding the studied WPT scenario is the time-varying nature of the wireless channel linking the WPD to the AP. We first investigate as a benchmark the offline case where the channel realizations are known non-causally prior to the starting of the horizon. For the offline case, by finding the optimal WPT duration and power allocations in the data transmission period, we derive an upper bound on the throughput of the WPD. We then focus on the online counterpart of the problem where the channel realizations are known causally. We prove that the optimal WPT duration obeys a time-dependent threshold form depending on the energy state of the WPD. In the subsequent data transmission stage, the optimal transmit power allocation for the WPD is shown to be of a fractional structure where at each time slot a fraction of energy depending on the current channel and a measure of future channel state expectations is allocated for data transmission. We numerically show that the online policy performs almost identical to the upper bound. We then consider a data sensing application, where the WPD adjusts the sensing resolution to balance between  the quality of the sensed data and the probability of successfully delivering it. We use Bayesian inference as a reinforcement learning method to provide a mean for the WPD in learning to balance the sensing resolution. We illustrate the benefits of the  Bayesian inference over the traditional approaches such as $\epsilon$-greedy algorithm using numerical evaluations. 

\end{abstract}

\section{Introduction}
\subsection{Motivation}
With the rapid increase in the number of battery-powered devices, energy harvesting (EH) technology provides a convenient window of opportunity to bypass the challenging, and in some cases infeasible task of replacing batteries. Traditional approaches in EH technologies harvest energy from natural resources such as wind, solar, etc. The inherent challenge of EH from natural resources is the stochastic nature of the EH process, which dictates the amount and availability of harvested energy that is beyond the control of system designers. Towards this end, wireless power transfer (WPT) \cite{6951347} is considered as a promising technology to provide the network administrators a leverage on replenishing the remote devices for proper network operations, by utilizing the RF signals as a mean to transfer power to WPDs.

In wireless powered communication networks (WPCNs) \cite{6678102,8048675,7866871}, WPT occurs in the down-link (DL) to replenish the battery of WPDs which in turn is used for information transmission (IT) in the up-link (UL). A fundamental question inherited in WPCNs is the optimum duration for WPT period and power allocation in the IT period. In \cite{6678102}, a hybrid access point (HAP) transmits power to a set of WPDs in the DL and collects their information using time-division-multiple-access (TDMA) in the UL. The aggregate throughput of the network is maximized by optimally determining the duration of WPT and UL time allocations for the WPDs. \cite{8048675} studies the effect of user cooperation in a two-user WPCN to maximize a weighted sum throughput by optimizing the beamforming vector in the WPT period, WPT duration, and data transmission durations. In \cite{7866871}, for a wireless powered underlay cognitive radio network in the presence of a primary user (PU), sum throughput of the secondary users (SUs) are maximized while meeting a constraint on the maximum average interference to the PU.

The above works, \cite{6678102,8048675,7866871}, perform a single-time-slot optimization assuming that the channel stays constant and all the harvested energy in a slot is totally used in the same time slot. Differently, \cite{7676282} assumes an infinite horizon throughput maximization problem where the harvested energy is allowed to be used in later times. It was shown that this strategy significantly improves the throughput albeit having high computational complexity. In \cite{7492928}, the half-duplex (HD) setting of \cite{7676282} is extended to consider a full-duplex (FD) scenario where the energy and information transfer occurs simultaneously and it was shown that the performance is significantly improved compared to the HD setting. \cite{7996351} studies a finite-horizon throughput maximization where a non-orthogonal multiple access (NOMA) is used to simultaneously decode information from multiple devices when the CSI is known non-causally. 

In the above works, it is assumed that in a single WPT instance, i.e., transmission of energy in the DL and reception of information in the UL, the channel state stays constant. However, in practice, this assumption is usually not valid. In this work, we study the problem of finite horizon throughput maximization, where both WPT and IT period is exposed to multiple random realizations of channel. The objective is to judiciously determine the optimal WPT duration and power allocations in the IT period. We emphasize that although the power allocation problem over finite horizon has been addressed previously, the joint problem of WPT duration and power allocation optimization has not been addressed.  The channel state information (CSI) is available causally and only in the IT period. The availability of causal CSI, makes the problem investigated here challenging, since any decision at any time slot has a cascading effect on the future outcomes. 

We first consider the offline problem by assuming that the CSI is available non-causally. For the offline case, we obtain closed form expressions to find the optimal WPT duration and power allocation in the IT period. We use the insights gained from the offline case, to develop an optimal online policy that maximizes the expected finite horizon throughput by optimally determining the WPT duration and power allocation in the IT period. Specifically, we formulate the problem of optimal WPT duration using the theory of stopping times. A stopping time is a random variable whose  value maximizes a certain property of interest in a stochastic process. We show that there exist a  time-dependent threshold on the energy level of the WPD in which it is optimal to stop WPT and start the IT period. Then, we show that the optimal power allocation in the IT period follows a fractional structure in which the WPD at each time slot allocates a fraction of its energy that depends on the current channel state as well as a specific measure of future channel expectations. 

We then consider a data sensing application, where the WPD is able to determine the sensing resolution of the data to be sent to the AP for further processing.  A high resolution data increases the performance of the application at the AP which depends on the resolution of the data transmitted by WPD. Meanwhile, a high resolution sensing setting produces more bits compared to a lower resolution setting which compromises the performance of the WPD in terms of successfully delivering the data. Therefore, an optimal sensing resolution is required to  balance the quality of the sensed data and the probability of successfully delivering it. We use Bayesian inference as a reinforcement learning method to provide a mean for the WPD in learning to balance the sensing resolution. We illustrate the benefits of the  Bayesian inference over the traditional approaches such as $\epsilon$-greedy algorithm using numerical evaluations.

\subsection{Contributions}
The contributions of the paper are summarized as follows:
\begin{itemize}
\item We formulate the problem of finite throughput maximization in a WPCN. We allow the finite horizon to span over multiple time slots where CSI changes randomly over time and it is available only causally at the transmitter during the IT period.
\item For the offline problem, where CSI is known non-causally, we derive a closed form expressions and enable a tractable framework to optimize both the WPT duration and power allocation in the IT period.
\item For the offline counterpart, we show that the optimal power allocations have a fractional structure depending on the current channel state as well as future channel states. 
\item Motivated by the results obtained from the offline problem, we formulate the online problem by assuming that the CSI is available only causally. 
\item We show that the optimal WPT duration for the online case has a time dependent threshold structure on the available energy of the WPD. We provide an easy to implement method to numerically calculate the thresholds.
\item Similar to the offline case, we show that the optimal power allocation for the online counterpart follows a fractional structure. The WPD allocates a fraction of its available energy in each time slot. Unlike the offline case, optimal fractions in the online case depends on the current channel state and a measure of the future channel state expectations.
\item We then consider a sensing application where the WPD uses a Bayesian framework to learn to determine the resolution of the sensing to balance the quality of the sensed data and the probability of successfully delivering it. We show that the Bayesian frame-work converges much faster, by judiciously exploring in the action space of the problem, than its classic counterpart $\epsilon$-greedy algorithm.
\end{itemize}
\subsection{Related Work}
\label{sec:related}

WPCN has been studied in the literature under different settings . \cite{7417596} studies a heterogeneous WPCN with the presence of EH and non-EH devices to find out how the presence
of non-harvesting nodes can be utilized to enhance the network
performance, compared to pure WPCNs. In \cite{7100855}, problem of throughput maximization in the presence of an EH relay is studied where the relay cooperatively help the source node in relaying its messages. The outage problem for a three node WPCN is analyzed in \cite{7018086,7248986} where both source and relay harvest energy for a certain duration, and then the source transmits to destination by using the relay. User cooperation is also studied in multiple works \cite{8048675,7037009,7500446} to improve the performance of the WPCN by exploiting the cooperative diversity. Multiple works also studied the WPCN in the context of cloud computing \cite{1708-08810,7417552,8234686,7997477}. Throughput maximization for WPCN is studied in \cite{6678102, 7676282,7492928,7996351}. Per time slot throughput maximization is studied in \cite{6678102}. By allowing the storage of the energy in a battery by the WPD, \cite{7676282} studies infinite horizon throughput maximization in HD mode and the results are extended to FD mode in \cite{7492928}. By adopting a NOMA strategy and under non-causal CSI, \cite{7996351} studies the problem of finite horizon throughput maximization.

Finite horizon throughput maximization has been extensively addressed from communication perspective in the literature for non-RF EH techniques. For example,  \cite{5992841} aims at maximizing the finite horizon throughput by dynamically adjusting the transmission power in an offline setting where CSI and the EH information (EHI) is non-causally available at the transmitter for the duration of the deadline. Packet transmission time minimization over a finite horizon with non-causal EHI and a static channel is studied in \cite{6094139}. However, in practice, the finite horizon spans over multiple time slots, and the CSI and EHI are not usually available. For time varying scenarios where EHI or CSI (or both) are available only causally, the problem needs to be solved dynamically. In \cite{6897968,7008488,7032105,7865904} under different EHI and CSI assumptions, the problem of finite horizon throughput maximization is formulated as a dynamic program (DP) and the optimal policy is evaluated by numerically solving the DP.  The solution is later stored in the devices as a {\em look-up table}. However, the DP solutions are computationally expensive, and they require large memory space to store the solutions, which is usually prohibitive for resource-constrained IoT devices. Moreover,  calculating and disseminating the optimal look up tables in a network consisting of large number of WPDs is inherently challenging and introduces large overheads \cite{7736112}. Finally, the complexity of the numerical solutions increase exponentially with respect to the number of states in the DP formulation. Recently, \cite{7959595} studied the problem of energy efficient scheduling for a non-RF EH over a finite horizon by developing a low complexity online heuristic policy that is built upon the offline solution and it can achieve close performance with respect to the offline policy. However, albeit the good performance, it is not evident how the algorithm would incorporate the optimal duration of the WPT period. 

In this work, we investigate the problem of finite horizon throughput maximization in a WPCN where an WPD harvests energy from WPT of the AP and then allocates the harvested energy in the subsequent time slots to transmit its data. Unlike the previous works, we consider a scenario where the CSI evolves randomly over the duration of the deadline, and CSI is only causally available at the transmitter which necessitates an online optimization framework. We avoid the complexity of the tabular methods (such as value iteration algorithm \cite{sutton}) by deriving closed form solutions for the optimal WPT duration and power allocations in the IT period. We show how the simple closed form expressions can be used to address a sensing application where the utility of the sensed data is captured by its resolution and probability of successfully delivering it. We address the problem in a reinforcement learning framework, where the optimum sensing resolution is learned by the WPD in a sequence of actions and observations. Finally, we conduct extensive simulations to verify our analytical findings.

\subsection{Outline}
The paper is organized as follows: In Section \ref{SystemModel}, we formally present the system model and all relevant assumptions. In Section \ref{PF}, we formulate the problem of finite horizon throughput maximization. In Section \ref{sec:optimal_off}, we provide an upper bound on the maximum achievable throughput by assuming non-causal information. In Section \ref{sec:optimal_on}, we solve the online counter-part of the problem by assuming only causal information. In Section \ref{sec:AL}, we address the sensing application and in Section \ref{numres}, we provide Monte-Carlo simulations to verify our findings. Finally, we conclude the paper in Section \ref{CON}.

\section{System Model}\label{SystemModel}
 \begin{figure*}[ht]
  \centering
    \includegraphics[scale=.5]{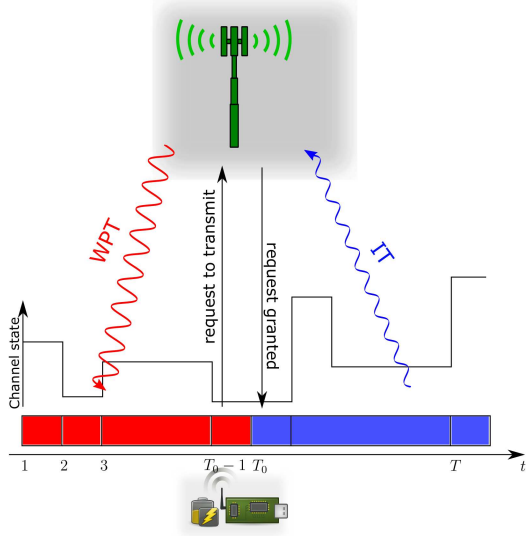}
		  \caption{System model.}
			\label{fig:offload}
\end{figure*}
We consider a point-to-point communication wireless channel where an WPD sends data to an AP by dynamically allocating power as shown in Figure \ref{fig:offload}. The AP uses WPT to replenish the battery of the WPD. The WPT and information transmission (IT) periods are non-overlapping in time, assuming a half-duplex transmission scenario. The WPD first harvests energy for a certain duration and then starts transmitting data to the AP. The duration of EH and IT periods is governed by the channel gain process which jointly affects the amounts of the harvested energy and transmitted data. We assume a discrete time scenario over a finite horizon. The time is slotted $t = 1,\ldots,T$ and $T<\infty$ denotes the frame length in units of slots. Let $g(t)$, $E_h(t)$ be the channel gain, and the amount of harvested energy at time slot $t$. Specifically, the amount of harvested energy at time slot $t$ is available at the beginning of slot $t+1$. The wireless channel is modeled as a multi state independent and identically distributed (iid) random process with $N$ levels.  The channel gain remains constant for a duration of a  time slot but changes  randomly from one time slot to another.  Let $g(t)\in \{g_1,\ldots,g_N\}$ be the channel power gain at slot $t$. We set $\mathds{P}(g(t)=g_n)=q_n$\footnote{Note that $g_n$'s can be obtained by discretizing a continuous time channel process.}. The WPD only has causal CSI and only during the IT period. 

The AP transmits a power beacon of $P$ watts over the wireless channel for a duration of $T_0-1$ time slots. The parameters that depend on the slot duration are normalize by the duration of a slot, and thus, we will refer to power and energy interchangeably. Assuming channel reciprocity, the amount of energy harvested by the WPD at time $t$ is $E_h(t)=\eta g(t)P$,
where $\eta$ is a constant representing the efficiency of the energy harvesting process\footnote{Note that $\eta$ in practice is a function of the received power and cannot be assumed to be a constant. However, assuming a variable $\eta$ does not change the results of the paper. Thus, for ease of presentation, we assume that $\eta$ is constant.}. The energy state of the WPD at time slot $t$ is denoted by $E(t)$. Let us denote $e_n = \eta g_n P$ as the amount of harvested energy when the channel state is at level $n$.

At time slot $t\geq T_0$, the WPD transmits with power $p(t)$, and the received power at the AP is $p(t)g(t)$.  In order to develop a tractable analytical solution, we assume a widely used empirical transmission energy model as in \cite{6574874,7880703,4674675,4357594,8403495,1801-03668}. Specifically, the instantaneous rate of transmitting with power $p(t)$ when the channel gain is $g(t)$ is calculated by
\begin{align}
r(t) =\sqrt[m]{\frac{p(t)g(t)}{\lambda}}  \label{eq:monomial}
\end{align}
where $\lambda$ denotes the energy coefficient incorporating the effects of bandwidth and noise power and $m$ is the monomial order determined by the adopted coding scheme \cite{1801-03668}. Figure \ref{fig:monomial} \cite{1801-03668}, compares the actual transmission rate with the monomial model described in \eqref{eq:monomial}. The approximated energy rate model, although may not be general for all cases, provides closed-form solutions for a challenging dynamic problem and gives insights to a practical and emerging problem.

 \begin{figure}[ht]
  \centering
    \includegraphics[scale=.15]{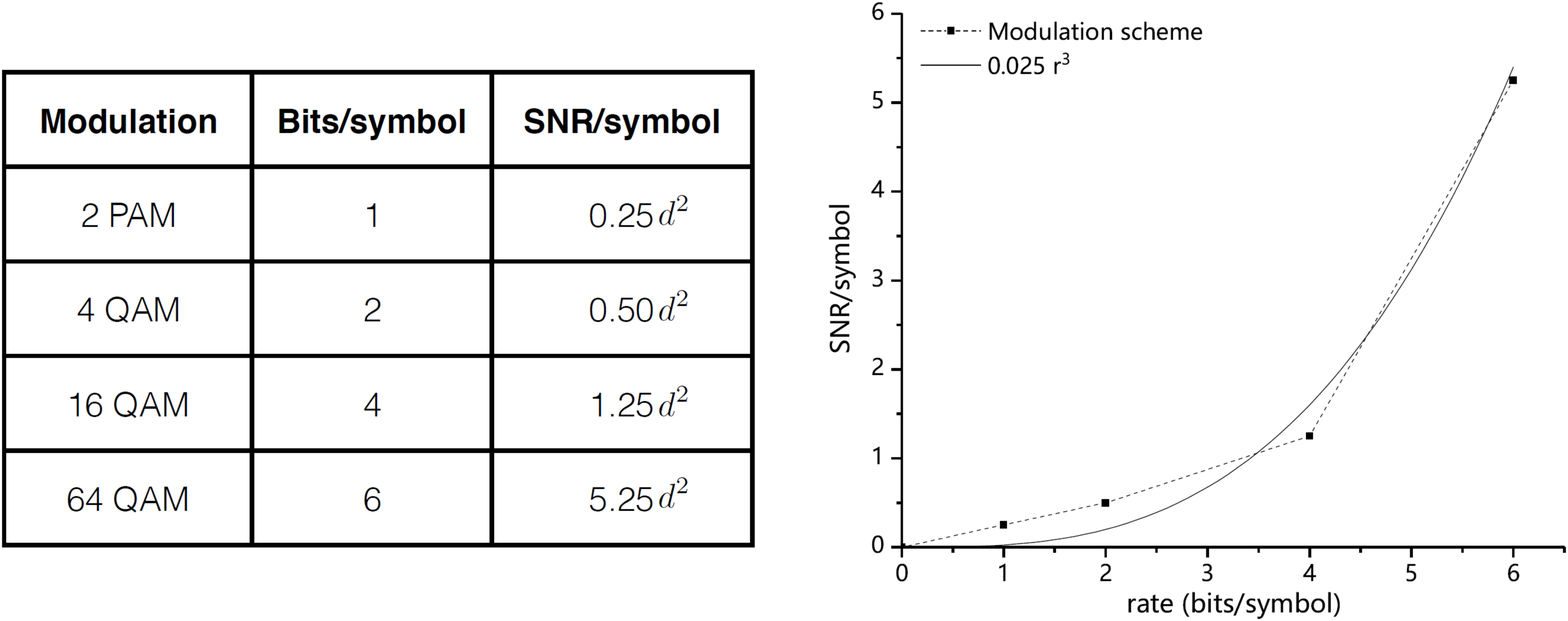}
		  \caption{The comparison of monomial and actual transmission rate and required signal-to-noise (SNR) ratio per symbol for $m=3$ and $\lambda = 0.025$ as given in \cite{1801-03668}. $d$ represents the minimum distance between signal points.}
			\label{fig:monomial}
\end{figure}

We consider the decentralized implementation. At the beginning of each slot, the WPD has the opportunity to inform the AP to stop WPT and begin IT period.  The objective is to judiciously determine the optimal WPT period duration, $T_0$,  and optimal power allocation in the IT period, $p(t)$ for $t=T_0,\ldots,T$, to maximize the finite horizon throughput. In Figure \ref{fig:Bevol}, we illustrate a sample realization of the battery of the WPD. The time frame has 10 time slots of $1ms$. The WPD accumulates energy until $t=2$. At $t=3$, since the available energy is larger than the threshold, the WPT period is stopped and the IT period began. 

\begin{figure}[ht]
  \centering
    \includegraphics[scale=.6]{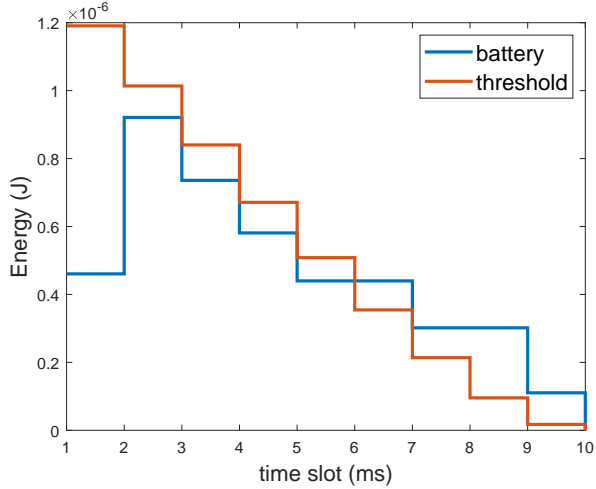}
		  \caption{An illustrative example of the battery evolution, $E(t)$, where $T=10$.}
			\label{fig:Bevol}
\end{figure}

\section{Problem Formulation}
\label{PF}

In this section, we formulate a joint optimization problem that finds the optimal trade-off between the EH and offloading periods, and the dynamic control of transmission power during the offloading period. 
%
More specifically, we aim at solving the following optimization problem.
\begin{align}
\max_{T_0,\{p(t)\}^{T}_{t=T_0}} &\sum^{T}_{t=T_0} \mathbf{\sqrt[m]{\frac{g(t)p(t)}{\lambda}}} \label{opt_problem}\\
&p(t)\leq E(t),\hspace{0.5cm} t =T_0,\ldots,T,\label{constraint1}\\
&E(t+1) = E(t) + E^h(t),\hspace{0.5cm} t =1,\ldots,T_0-1 \label{constraint2}\\
&E(t+1) = E(t)-p(t),\hspace{0.5cm} t =T_0,\ldots,T.\label{constraint3}
\end{align}
Note that the objective function \eqref{opt_problem} is the total number of transmitted bits in the offloading period, \eqref{constraint1} ensures that the consumed energy does not exceed the available energy, \eqref{constraint2} and \eqref{constraint3} are the battery dynamics in the EH and offloading periods, respectively.  We first solve the offline version of the optimization problem by assuming that the channel gains are available prior to the optimization. Using the insights from the offline problem, we will design an optimal online policy, where the channel gains are only available causally.

\section{Optimal Offline Policy}
\label{sec:optimal_off}
We consider the offline counterpart of the optimization problem in (\ref{opt_problem}). Thus, we assume that values of $g(t)$ are known non-causally for $t = 1,\ldots,T$. Assuming that the optimal value of $T_0$ is given, we first aim at optimizing the power allocation in the IT period. We are interested in maximizing the following function

\begin{align}
\max_{p(t)} & \sum^{T}_{t=T_0} r(t) \nonumber\\
&0\leq p(t)\leq E(t).\nonumber
\end{align}

In Theorem \ref{off_optimal}, we show that the optimal policy, that maximizes the total number of bits transmitted bits in the IT period, allocates at each time slot a fraction of the available energy which depends on the channel realizations.




\begin{theorem}
\label{off_optimal}
For a given $T_0$ and realizations of $g(t)$ for $t=1,\ldots,T$, the optimal dynamic power allocation for the offline problem is calculated by 
\begin{align}
p^*(t) = \frac{g(t)^{\frac{1}{m-1}}}{g(t)^{\frac{1}{m-1}} + G(t+1)^{\frac{1}{m-1}}} E(t)\label{off_alfa}
\end{align}
where
\begin{align}
G(t)=&\left\{
\begin{array}{ll}
\left[g(t)^{\frac{1}{m-1}} + G(t+1)^{\frac{1}{m-1}}\right]^{m-1},& \text{if}\   t\leq T\\
0,& \text{if}\  t> T
\end{array}
\right.,\label{I}
\end{align}
and the maximum number of transmitted bits is calculated as
\begin{align}
\sum^T_{t=T_0} r^*(t) = \sqrt[m]{\frac{E(T_0)}{\lambda} G(T_0)}
\end{align}
\end{theorem}
\begin{proof}
The proof is given in Appendix \ref{proof:offline}.
\end{proof}
Thus, the offline optimization problem becomes:
\begin{align}
\max_{T_0} & \sqrt[m]{\frac{E(T_0)}{\lambda} G(T_0)} \\
&2\leq T_0\leq T.\nonumber
\end{align}

The above maximization problem has only one integer variable and hence, the optimal value for $T_0$ can be easily calculated numerically.

\section{Optimal Online Policy}
\label{sec:optimal_on}
 Note that, in the online case, $g(t)$ is only available causally. Therefore, the optimization problem in \eqref{opt_problem}-\eqref{constraint3} cannot be solved using offline optimization tools and an online algorithm is required for its solution. A common approach to solve similar problems is to use dynamic programming (DP) to find the solution numerically, and store the optimal decisions in a look-up table for the WPD. However, solving a DP and storing the result is prohibitive for resource constrained WPDs. In the following, we extend the insights gained in the offline case to the online counterpart of the optimization problem in \eqref{opt_problem}. 

At each time slot $t\geq T_0$, the WPD allocates a fraction of its remaining energy and allocates $p(t)=\alpha(t)E(t)$ as its transmit power. Hence, the optimization problem converts to:
\begin{align}
\max_{T_0,\{\alpha(t)\}^{T}_{t=T_0}} &\sum^{T}_{t=T_0} \mathbf{\sqrt[m]{\frac{g(t)\alpha(t)E(t)}{\lambda}}} \label{on_problem}\\
&0\leq\alpha(t)\leq 1,\hspace{0.5cm} t =T_0,\ldots,T,\label{on_constraint1}\\
&E(t+1) = E(t) + E^h(t),\hspace{0.5cm} t =1,\ldots,T_0-1 \label{on_constraint2}\\
&E(t+1) = (1-\alpha(t))E(t),\hspace{0.5cm} t =T_0,\ldots,T.\label{on_constraint3}
\end{align}


\subsection{Dynamic Energy Allocation}
In this section, we first optimize the values of $\alpha(t)$ by conditioning on $T_0$. Then using the obtained result, we will give a criteria for stopping the EH process, i.e., optimizing the value of $T_0$. 

Let the offloading period begin at $T_0$ and aim to maximize the throughput over $T-T_0$ time slots by using DP. The problem is recursively solved starting at the last time slot $T$, and the result is propagated by recursion until it reaches $t=T_0$.  We denote the instantaneous reward of choosing $\alpha(t)$ by $U_{\alpha(t)}(E(t),g(t)) $ which is the instantaneous number of bits transmitted to the AP, when the the amount of available energy at time $t$, is $E(t)$ and the channel power gain is at state $g(t)$. Thus,
\begin{align}
U_{\alpha(t)}(E(t),g(t)) = \sqrt[m]{\frac{\alpha(t)g(t)E(t)}{\lambda}}.
\end{align}

We denote the action-value function by $V_{\alpha}(E(t),g(t))$ which is equal to the instantaneous reward of choosing $\alpha(t)$ plus the expected number of bits that can be transmitted in the future. Hence, the action-value function evolves as,
\begin{align}
V_{\alpha(t)}(E(t),g(t))= U_{\alpha(t)}(E(t),g(t)) + \sum q_i V(E(t+1),g_i), 
\end{align}
where, $V(E(t),g(t))$ is the value function defined as,
\begin{align}
V(E(t),g(t)) =\max_{\alpha(t)} V_{\alpha(t)}(E(t),g(t)).
\end{align}

Note that at the last time slot, i.e., $t=T$, all the energy in the battery will be used for transmission, i.e., $\alpha(T) = 1$.  Thus, it follows that, 
\begin{align}
V(E(T),g(t))=&U_1(E(T),g(T)) \nonumber\\
=& \sqrt[m]{\frac{g(T)E(T)}{\lambda}}\nonumber\\
=&\sqrt[m]{\frac{g(T)(1-\alpha(T-1))E(T-1)}{\lambda}}.
\end{align}


We maximize the action-value function at $t=T-1$ by optimizing $\alpha(T-1)$ as follows,
\begin{align}
V_{\alpha}(E(T-1),g(T-1)) =& U_{\alpha}(E(T-1),g(T-1)) + \sum q_i V((1-\alpha(T-1))E(T-1),g_i) \nonumber\\
=&\sqrt[m]{\frac{g(T-1)\alpha(T-1) E(T-1)}{\lambda}} + \sum q_i \sqrt[m]{\frac{g_i((1-\alpha(T-1))E(T-1))}{\lambda}}.\label{V_alpha_T_1}
\end{align}
It is easy to see that (\ref{V_alpha_T_1}) is concave with respect to $\alpha(T-1)$. Therefore, by differentiating (\ref{V_alpha_T_1}), the optimal $\alpha(T-1)$ can be calculated as follows:
\begin{align}
\alpha^*(T-1) = \frac{g(T-1)^{\frac{1}{m-1}}}{g(T-1)^{\frac{1}{m-1}}+Q(T-1)^{\frac{m}{m-1}}}\label{alfa1},
\end{align}
where,
\begin{align}
Q(T-1) = \sum q_i\sqrt[m]{g_i}.\label{Q1}
\end{align}
The corresponding value function can also be calculated as
\begin{align}
&V(E(T-1),g(T-1))= \sqrt[m]{\frac{E(T-1)}{\lambda}}\big(g(T-1)^{\frac{1}{m-1}} + Q(T-1)^{\frac{m}{m-1}}\big)^\frac{m-1}{m}.\label{V1}
\end{align}

In a similar manner as above, we can recursively calculate the optimal $\alpha(t)$ for $t=T-2,\ldots,T_0$. The result is summarized in the following theorem. 

\begin{theorem}
\label{optimal}
For any $t = T-1, \ldots,T_0$, the optimal decision is to choose
\begin{align}
\alpha^*(t) = \frac{g(t)^{\frac{1}{m-1}}}{g(t)^{\frac{1}{m-1}}+Q(t)^{\frac{m}{m-1}}},\label{alfa_k}
\end{align}
where
\begin{align}
Q(t) = \sum q_i \big(g_i^{\frac{1}{m-1}} + Q(t+1)^{\frac{m}{m-1}}\big)^{\frac{m-1}{m}} \label{Q_k}.
\end{align}
The corresponding value function is 
\begin{align}
V(E(t),g(t)) = \sqrt[m]{\frac{E(t)}{\lambda}}\big(g(t)^{\frac{1}{m-1}} + Q(t )^{\frac{m}{m-1}}\big)^\frac{m-1}{m}\label{V_k}
\end{align}
\end{theorem}
\begin{proof}
The proof is given in Appendix \ref{proof:optimal}.
\end{proof}

Theorem \ref{optimal} gives a framework to dynamically allocate energy at each time slot $t\geq T_0$. Instead of numerically solving the DP and storing it in a large look up table, WPD needs to just calculate and store an array of values with a maximum dimension of $T$. The closed form expressions derived in (\ref{alfa_k})-(\ref{V_k}) significantly simplify the procedure to optimize $T_0$. We will use these results to find an structure for the optimal stopping time problem in the subsequent section.

\subsection{Optimal Stopping time for the EH Process}

In the following, we derive the optimal stopping time for the EH process, i.e., optimizing $T_0$ in \eqref{opt_problem}-\eqref{constraint3}. Recall that the WPD accumulates energy up to some time $t$, and then stops the EH process to start offloading its task. Also, recall that during EH, the WPD is blind to the channel conditions. If the WPD stops the EH process at time $t$, then the expected number of bits that can be transmitted is
\begin{align}
\sum q_i V(E(t),g_i) &= \sum q_i \sqrt[m]{\frac{E(t)}{\lambda}}\big(g_i^{\frac{1}{m-1}} + Q(t)^{\frac{m}{m-1}}\big)^\frac{m-1}{m}\nonumber\\
& = \sqrt[m]{\frac{E(t)}{\lambda}}Q(t-1). \label{V_avg}
\end{align}
Note that  (\ref{V_avg}) follows from the definition of $Q(t)$ given in (\ref{Q_k}).

Let $J_t(E(t))$, $t=1,\ldots,T$ be the maximum expected number of bits that can be transmitted if the EH process is stopped at time $t$, and the amount of available energy is $E(t)$. At any time $t$, the WPD will either stop or continue the EH process. The optimal stopping time for the EH process can be formulated as
\begin{align}
\max_{t\leq T}\,\, J_t(E(t)),
\end{align}
where
\begin{align}
J_t(E(t)) = \max\left(\sqrt[m]{\frac{E(t)}{\lambda}}Q(t-1), \mathbb{E}(J_{t+1}(E(t+1))\bigg|E(t))\right).
\end{align}

The problem can be formulated as a DP and recursively solved for every possible $E(t)$ and $t$. Before proceeding, we need the following lemma.
\begin{lemma}
\label{lemma:Q_decreasing}
$Q(t)$, defined in (\ref{Q_k}) is a monotonically decreasing function in $t$.
\end{lemma}
\begin{proof}
\begin{align}
\frac{Q(t)}{Q(t+1)} &= \frac{\sum q_i \big(g_i^{\frac{1}{m-1}} + Q(t+1)^{\frac{m}{m-1}}\big)^{\frac{m-1}{m}} }{Q(t+1)}\nonumber\\
&=\sum q_i \big( 1+ \frac{g_i^{\frac{1}{m-1}}}{Q(t+1)^{\frac{m}{m-1}}}  \big)^{\frac{m-1}{m}}>1.
\end{align}
It readily follows that $Q(t)>Q(t+1)$.
\end{proof}

Note that at $t=T$, the best strategy is to stop the EH process and start offloading, since otherwise no bits can be offloaded to the AP. Thus,
\begin{align}
J_T(E(T)) = \sqrt[m]{\frac{E(T)}{\lambda}} Q(T-1).
\end{align}

We continue the recursive evaluation at time slot $t =T-1$. We have,
\begin{align}
&J_{T-1}(E(T-1)) \nonumber\\
&= \max(\sqrt[m]{\frac{E(T-1)}{\lambda}} Q(T-2),\mathbb{E}(J_T(E(T))|E(T-1)))\nonumber\\
&= \max(\sqrt[m]{\frac{E(T-1)}{\lambda}} Q(T-2),\sum q_i\sqrt[m]{\frac{E(T-1)+e_i}{\lambda}} Q(T-1))
\end{align}
Since $Q(T-2)>Q(T-1)$ as proven in Lemma \ref{lemma:Q_decreasing}, if $E(T-1)\geq \gamma(T-1)$ , then 
\begin{align}
\sqrt[m]{\frac{E(T-1)}{\lambda}} Q(T-2)\geq \sum q_i\sqrt[m]{\frac{E(T-1)+e_i}{\lambda}} Q(T-1)),
\end{align}
where $\gamma(T-1)$ is the solution to the following equation
\begin{align}
\sum q_i \sqrt[m]{1+\frac{e_i}{\gamma(T-1)}} = \frac{Q(T-2)}{Q(T-1)}. \label{gamma_T_1}
\end{align}
Note that $\gamma(T-1)$ admits a unique solution because the left hand side of (\ref{gamma_T_1}) is a strictly decreasing function in $\gamma(T-1)$ and its range belongs to $(1,\ \infty)$. Also, from Lemma \ref{lemma:Q_decreasing}, we know that $\frac{Q(T-2)}{Q(T-1)}>1$. 
Hence, it is optimal to stop the EH process at time $T-1$ if $E(T-1)\geq \gamma(T-1)$. This suggests that the optimal stopping times are governed by a time varying threshold type structure, where at any given time $t$, it is optimal to stop the EH process if $E(t)\geq\gamma(t)$. Before, proving this observation, we need the following lemma.
\begin{lemma}
\label{lemma:Qfrac_decreasing}
For any $k=1,\ldots,T-1$, we have
\begin{align}
\frac{Q(T-k-1)}{Q(T-k)}<\frac{Q(T-k)}{Q(T-k+1)}
\end{align}
\end{lemma}
\begin{proof}
By using (\ref{Q_k}), we have
\begin{align}
\frac{Q(T-k-1)}{Q(T-k)} &= \frac{\sum q_i \big(g_i^{\frac{1}{m-1}} + Q(T-k)^{\frac{m}{m-1}}\big)^{\frac{m-1}{m}} }{Q(T-k)}\nonumber\\
&=\sum q_i \big(1 + \frac{g_i^{\frac{1}{m-1}}}{Q(T-k)^{\frac{m}{m-1}}} \big)^{\frac{m-1}{m}},
\end{align}
and,
\begin{align}
\frac{Q(T-k)}{Q(T-k+1)} &= \frac{\sum q_i \big(g_i^{\frac{1}{m-1}} + Q(T-k+1)^{\frac{m}{m-1}}\big)^{\frac{m-1}{m}} }{Q(T-k+1)}\nonumber\\
&=\sum q_i \big(1 + \frac{g_i^{\frac{1}{m-1}}}{Q(T-k+1)^{\frac{m}{m-1}}} \big)^{\frac{m-1}{m}}.
\end{align}
From Lemma \ref{lemma:Q_decreasing}, we have $Q(T-k)>Q(T-k+1)$ and thus the lemma holds.
\end{proof}

In the following theorem, we give the structure of the optimal stopping policy.

\begin{theorem}
\label{theorem:threshold}
At each time slot $t$, the optimal decision is to stop the EH process if $E(t)\geq \gamma(t)$, where $\gamma(t)$ is the solution to the following equation,
\begin{align}
\sum q_i \sqrt[m]{1+\frac{e_i}{\gamma(t)}} = \frac{Q(t-1)}{Q(t)}\label{gamma}
\end{align}
\end{theorem}
\begin{proof}
The proof is in Appendix \ref{proof_theorem:threshold}.
\end{proof}
The results established in Theorem \ref{optimal} and \ref{theorem:threshold} enables us to develop an online low complexity optimal algorithm that maximizes the expected throughput. The procedure is summarized in Algorithm \ref{alg}.

\begin{algorithm}
\caption{Online policy}\label{alg}
\begin{algorithmic}[1]
\State Initialize $Q(t)$ for $t= 0,\ldots,T-1$ using (\ref{Q_k}),
\State Initialize $\gamma(t)$ for $t = 1,\ldots,T-1$ using (\ref{gamma}),
\For{$t = 1:T$}
\If{$E(t)<\gamma(t)$}
\State continue the EH process
\Else 
\State $T_0 = t$,
\State Stop the  EH process,
\State Break
\EndIf
\EndFor
\For{$t = T_0:T$}
\State Calculate $\alpha(t)$ using (\ref{alfa_k}),
\State Transmit using $\alpha(t)E(t)$.
\EndFor
\end{algorithmic}
\end{algorithm}

\section{Optimal Sensing}
\label{sec:AL}
Thus far, we have considered a general problem of maximizing the throughput of a WPD over a finite horizon. However, in practice most WPDs will be sensors collecting data to provide situational awareness to the users \cite{wearable}. Hence, we next consider a related problem formulation, where we aim to maximize the overall utility of the system defined as a function of bits delivered successfully by the end of time horizon. For example, the WPD could  be an image sensor that transmits images to the AP tracking the eye movement, i.e., estimating the \emph{gaze} location of a person \cite{ishadow}. The accuracy of estimating the gaze depends on the number of pixels per frame. A gaze error varies from  $10-15$ pixels at $77$ pixels/frame to $0-3$ pixels at $1984$ pixels/frame \cite{Rostami}. Hence, providing a high resolution input data provides a higher utility. However, the increased utility in the application layer comes at a price of reduced chance of delivering the input data to the AP due to the finite time horizon and the dynamic nature of the wireless link. Hence, there exists an optimal trade-off in balancing the quality of input data and probability of delivering it successfully to the AP for processing.

The WPD is able to provide an input data at $K$ different \emph{quality} points where each quality point is described by the number of bits used, i.e., $L_k$ bits for $k = 1,\ldots,K$. Let the energy consumed for generating a type $k$ input data be $\mathcal{E}_k$ joules. If the WPD successfully delivers a type $k$ data, it receives a known utility of $Z(L_k)$, and zero otherwise. 
In other words, if the throughput of a transmission frame is larger than the packet size, then the packet is successfully delivered to the AP. We formulate the utility optimization as follows:

\begin{align}
\max_{L_k,T_0,\{\alpha(t)\}^{T}_{t=T_0}} &Z(L_k)\mathds{P}\Bigg(\sum^{T}_{t=T_0} \mathbf{\sqrt[m]{\frac{g(t)\alpha(t)E(t)}{\lambda}}}>L_k\Bigg) \label{extended_problem}\\
&0\leq\alpha(t)\leq 1,\hspace{0.5cm} t =T_0,\ldots,T,\label{const:1}\\
&E(t+1) = E(t) + E^h(t),\hspace{0.5cm} t =1,\ldots,T_0-1, \label{const:2}\\
&E(t+1) = (1-\alpha(t))E(t),\hspace{0.5cm} t =T_0,\ldots,T,\label{const:3}\\
&L_{min}\leq L_k \leq L_{max}.\label{const:4}
\end{align}
The above optimization problem consists of three sub-problems; choosing the size of the input data $L_k$, determining the optimal WPT duration $T_0$, and optimal power allocations in the IT period $\alpha(t)$, $t=T_0,\ldots,T$. The online policy presented in Algorithm \ref{alg} has the maximum possible throughput on the average with respect to any other policy. Hence, it has a better probability of success compared to any alternative policy. Thus, the WPD uses Algorithm \ref{alg} to determine the optimal duration of WPT and power allocations in the IT period. It should be noted that the thresholds, that are used to determine the optimal WPT duration, need to be adjusted by adding $\mathcal{E}_k$ to all the thresholds to account for the energy consumption due to sensing the data. 

Let the event of successfully delivering a packet of $L_k$ bits be $\chi_k$. More specifically:
\begin{align}
\chi_k=\Bigg\{
\begin{array}{ll}
         1 & \mbox{if $\sum^{T}_{t=T_0} r(t)>L_k$},\\
        0  & \mbox{otherwise}.
        \end{array}\label{eq:pdp}
\end{align}

We rewrite the optimization problem of interest as follows: 
\begin{align}
\max_{\{L_k\}^{K}_{k=1}} Z(L_k) \mathds{E}(\chi_k) \label{opt:AL}
\end{align}
The WPD in the beginning of each transmission frame chooses a $L_k$ that optimizes the above optimization problem\footnote{Note that a better strategy is to choose the size of the data after observing the amount of harvested energy and the duration of IT period. However, this comes at the cost of conditioning on a two dimensional state which extremely complicates the solution that is not suitable for less capable WPDs.}. The unknown quantities in the optimization problem are $\mathds{E}(\chi_k)$, $k=1,\ldots,K$. We aim to learn these quantities using a reinforcement learning (RL) technique. The RL framework interacts with the environment and learns the values of the parameters of interest by observing the outcomes of its decisions. This problem can be efficiently formulated in the context of \emph{multi armed bandit} (MAB) problem. The parameters of interest in the MAB are denoted by $\theta_k=\mathds{E}(\chi_k) = \mathds{P}(\chi_k=1)$. We aim to efficiently infer each $\theta_k$ by interacting with the environment and observing the outcomes. In a MAB there are multiple arms (i.e., actions) each generating a random reward according to a probability distribution function (PDF). An agent sequentially chooses an action $x_t = k$ for $t=1,\ldots$ and readjusts it strategy by observing the reward with the hope of maximizing its expected reward. In our problem, there are $K$ actions. The WPD keeps  initial estimates of $\hat{\theta}_k$ about the unknown parameters $\theta_k$. The WPD chooses an action $x_t=k$ and observes the event  $Z(L_k)\cdot \chi_k$. Based on the observation, it updates $\hat{\theta}_k$ until the algorithm converges to the optimal value. The typical method for optimizing a MAB problem is by the well known $\epsilon$-greedy algorithm presented in Algorithm \ref{alg:greedy}. The $\epsilon$-greedy algorithm consists of two steps; exploration and exploitation. Exploration improves the estimate of non-greedy actions' values while exploitation is favorable when we reach a sufficient knowledge about the estimate of actions. $\epsilon$-greedy algorithm, with probability (w.p.) $1-\epsilon$, greedily chooses an action $k$ that maximizes $Z(L_k)\hat{\theta}_k$ and w.p. $\epsilon$ randomly chooses an action. In other words, w.p. $\epsilon$ the algorithm explores in the action space of the MAB while w.p. $1-\epsilon$ the algorithm exploits what it already knows. Although such an approach is guaranteed to approach the optimal performance \cite{sutton}, provided that $\epsilon$ is sufficiently small, the convergence rate of the algorithm is poor. This is because $\epsilon$-greedy algorithm does not judiciously explore in the parameter space. 
To speed up the convergence, we use a Bayesian inference method to judiciously explore in the action space of the MAB problem.  The augmentation of the Bayesian framework in MAB is known as Thompson sampling (TS) \cite{arxiv:TS}. To see how TS works, let us model the uncertainty $\theta_k$ by assuming a prior distribution for it.  Each $\theta_k$ is distributed according to a Beta distribution with parameters $a_k$ and $b_k$. In particular, for each arm $k$, the prior probability density function of $\theta_k$ is:
\begin{align}
\mathds{P}(\theta_k) = \frac{\Gamma(a_k+b_k)}{\Gamma(a_k)\Gamma(b_k)}\theta_k^{a_k-1}(1-\theta_k)^{b_k-1},
\end{align}
where $\Gamma(.)$ denotes the gamma function. The reason for choosing Beta as prior distribution is the conjugacy property of Beta distribution with Bernoulli distribution. In other words, if prior is Beta distributed and the likelihood is Bernoulli distributed, then the posterior distribution is also Beta distributed. This facilitates the process of sampling from the posterior distribution\footnote{Note that the conjugacy property only makes it easier to sample from the posterior distribution. In case where the posterior distribution does not admit any known PDF, efficient Monte-Carlo methods such as Markov chain Monte-Carlo (MCMC) \cite{MCMC} method and its variants such as Gibbs sampling can be used to efficiently sample from the posterior.}. Given a sample realization of $\chi_k$, we are interested in updating the posterior distribution of $\theta_k$. We have:
\begin{align}
\mathds{P}(\theta_k|\chi_k)&\propto \mathds{P}(\theta_k)\mathds{P}(\chi_k|\theta_k)\nonumber\\
&=\frac{\theta_k^{a_k-1}(1-\theta_k)^{b_k-1}}{B(a_k,b_k)}\theta_k^{\chi_k}(1-\theta_k)^{1-\chi_k}\nonumber\\
&\propto \theta_k^{a_k-1+\chi_k}(1-\theta_k)^{b_k-1+1-\chi_k}
\end{align}
Hence, the posterior distribution is also Beta distributed with parameters, $a_k + \mathds{1}_{\{\chi_k=1\}}$ and $b_k + \mathds{1}_{\{\chi_k=0\}}$. Note that at every given time, only a single observation regarding the chosen action is revealed. Hence, after retrieving the observation about an action, the parameters of the  posterior distribution is updated as:
\begin{align}
(a_k,b_k)\leftarrow\Bigg\{
\begin{array}{ll}
         (a_k,b_k) & \mbox{if $x_t \neq k$},\\
        (a_k + \chi_k,b_k+1-\chi_k) & \mbox{if $x_t = k$}.
        \end{array}\label{eq:update}
\end{align}

The TS algorithm is given in Algorithm \ref{alg:MAB_TS}.  Note that the only difference between the TS and $\epsilon$-greedy algorithms in the exploration phase of the problem. TS judiciously explores by modeling the uncertainty of each action using a distribution with decreasing variance in the number of observations explored. This prevents the TS from exploring the actions that are believed to be sub-optimal. Meanwhile $\epsilon$-greedy explores  the action space randomly, reducing the efficiency of the exploration phase.

\begin{algorithm}
\caption{$\epsilon$-greedy}\label{alg:greedy}
\begin{algorithmic}[1]
\For{$t = 1,2,\ldots$}
\State With probability $\epsilon$
\For{k=1,\ldots,K}
\State  $\hat{\theta}_k=\frac{a_k}{a_k + b_k}$
\EndFor
\State $
x_t\leftarrow\Bigg\{
\begin{array}{ll}
         \arg\max_k Z(L_k)\hat{\theta}_k & \mbox{with probability $1-\epsilon$},\\
        \mbox{choose a random action} & \mbox{with probability $\epsilon$} .
        \end{array}
$
\State Apply $x_t$ and observe $\chi_k$
\State update the posterior using \eqref{eq:update}
\EndFor
\end{algorithmic}
\end{algorithm}

\begin{algorithm}
\caption{TS}\label{alg:MAB_TS}
\begin{algorithmic}[1]
\For{$t = 1,2,\ldots$}
\State Sample from the posterior
\For{k=1,\ldots,K}
\State Sample $\hat{\theta}_k\sim beta(a_k,b_k)$
\EndFor
\State $x_t\leftarrow\arg\max_k Z(L_k)\hat{\theta}_k$
\State Apply $x_t$ and observe $\chi_k$
\State update the posterior using \eqref{eq:update}
\EndFor
\end{algorithmic}
\end{algorithm}

\section{Numerical Results}
\label{numres}
In this section, we compare the performance of the optimal online policy with that of the offline as well as two benchmark policies, namely \emph{uniform} and \emph{power-halving} policies. In uniform policy, the amount of harvested energy is uniformly distributed in the IT period. Power-halving policy allocates half of its available energy in each time slot in the IT period. The WPT duration for both uniform and power-halving policy is optimized using exhaustive search method. We also evaluate the performance of TS algorithm in the utility maximization problem developed in Section \ref{sec:AL} and compare it with that of $\epsilon$-greedy.

For the channel state, we assume two different channel models based on Rayleigh and Gilbert-Elliot (G-E) fading models. For Rayleigh fading, we assume an average channel gain of $1$. For G-E model, we assume that there are two state; good and bad. The gain of good state is $1$ and that of bad state is $0$. The good and bad states occur with probability of $0.6$ and $0.4$, respectively. We assume that the AP transmits with power $P=20$dBm which is normalized with respect to distance and EH efficiency.  Time slot duration is $1$ms, the bandwidth is assumed to be $2$KHz, and the noise power density is $176$ dBm/Hz.

\subsection{Rate-Energy Trade-off}
We first evaluate the rate-energy trade-off of the online policy which is the expected total number of bits transmitted with respect to the amount of harvested energy in a finite duration of $T$. In Figure \ref{fig:r-e_N}, for different values of channel discretization level, $N$, and a frame length of $15$ time slots, the rate-energy trade-off is depicted. For different values of $T$, and $N = 15$, Figure \ref{fig:r-e_T}, illustrates the rate-energy trade-off. We observe from the figures that, spending too much time for transmitting more energy in the EH period reduces the time for IT period which in turn reduces the throughput. On the other hand, if we reduce the EH period, there would be less energy in the IT period resulting in a reduced throughput. Hence, an optimal balance is required.

\begin{figure}
        \centering
        \begin{subfigure}[h]{0.45\textwidth}
                \includegraphics[width=\textwidth]{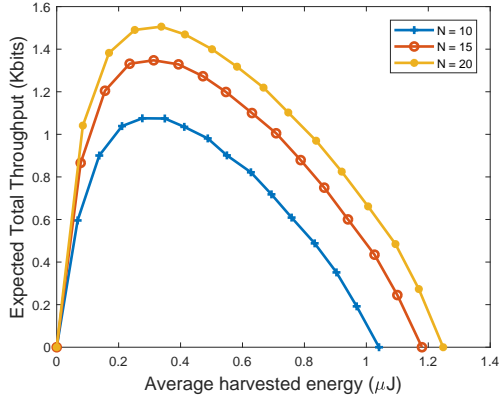}
                \caption{Expected throughput with respect to $N$.}
                \label{fig:r-e_N}
        \end{subfigure}
        \hfill
        \begin{subfigure}[h]{0.45\textwidth}
                \includegraphics[width=\textwidth]{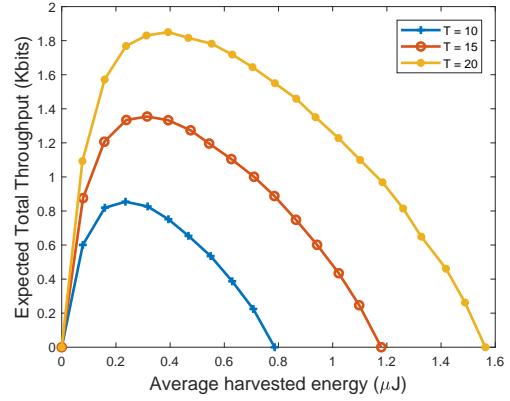}
                \caption{Expected throughput with respect to $T$. }
                \label{fig:r-e_T}
        \end{subfigure}
				\caption{The effect of channel discretization and deadline duration on the expected throughput.}
\label{R-E}
\end{figure}

\subsection{Performance Evaluation}
In Figure \ref{fig:rayl_N}, when the fading is Rayleigh, the expected total number of bits that are transmitted in $100$ time slots is depicted with respect to the number of channel discretization levels, $N$. We observe that as the number of channel levels increases, the discretization error decreases and hence the throughput of the all policies improve. The online policy achieves a throughput close to the upper-bound by optimally determining the WPT duration and power allocation in the IT period. Although the uniform and power-halving policies harvest energy for an optimum duration, they considerably perform poor due to the blind power allocation in the IT period.

\begin{figure}[ht]
  \centering
    \includegraphics[scale=.5]{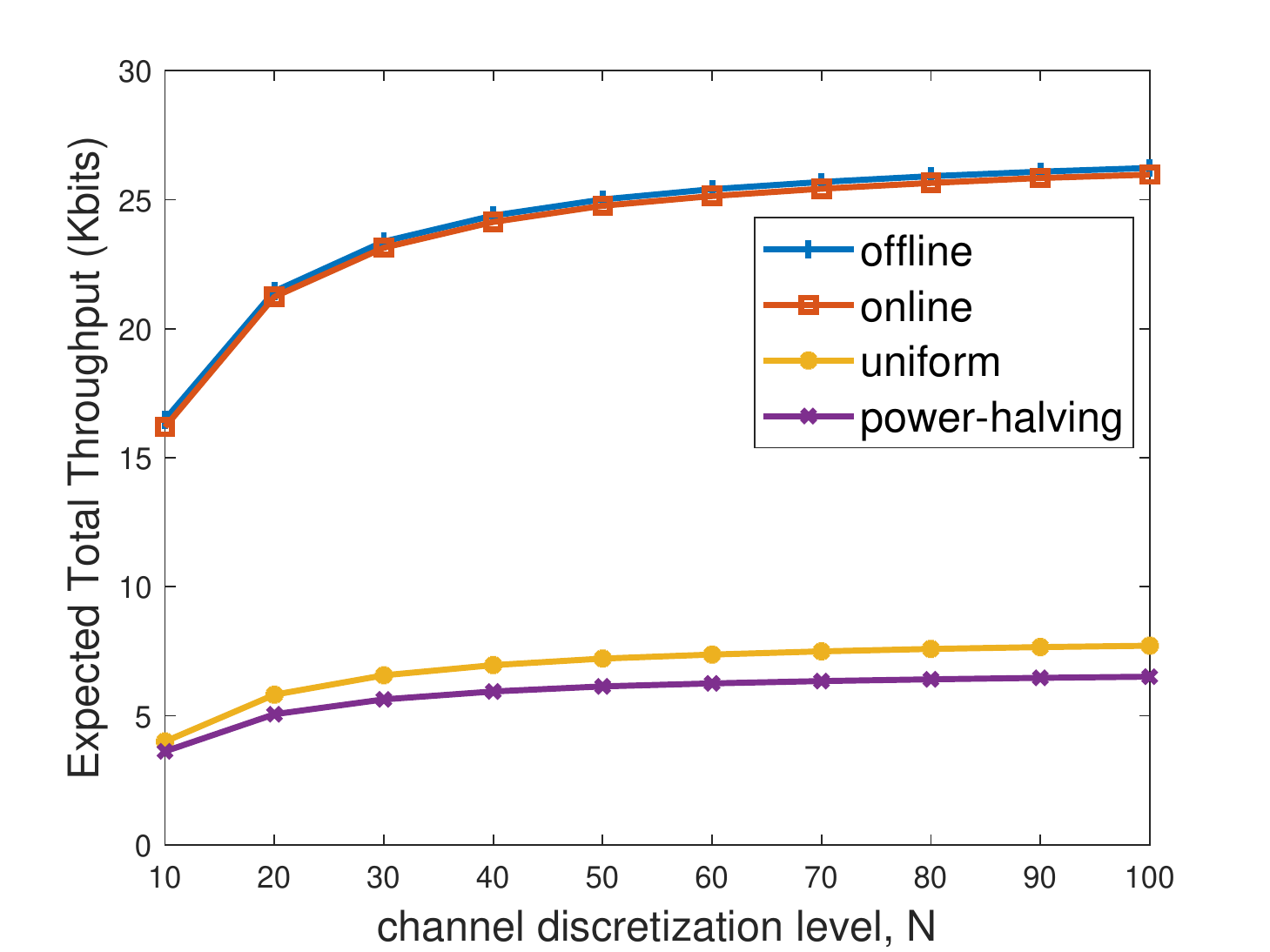}
		  \caption{Expected total throughput of the WPD with respect to the number of channel discretization levels in $T = 100$ time slots.}
			\label{fig:rayl_N}
\end{figure}
Next, we plot the expected total throughput of the WPD under Rayleigh and G-E fading models in Figure \ref{fig:rayl_T} and Figure \ref{fig:g-e_T}, respectively. Again, the online policy, for all values of $T$, achieves an outstanding performance compared to the offline policies. For smaller values of $T$, the power-halving policy achieves a good performance. However, as $T$ increases, due to the concave nature of the rate-power function, the power-halving strategy becomes significantly inefficient. On the other hand, uniform policy is able to perform better, for larger values of $T$, with respect to power-halving policy by allocating the harvested energy uniformly across the IT period.
\begin{figure}
        \centering
        \begin{subfigure}[h]{0.45\textwidth}
                \includegraphics[width=\textwidth]{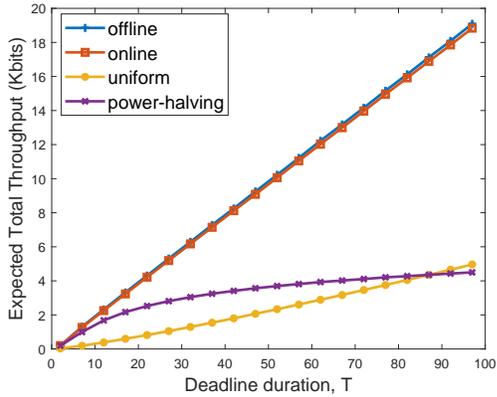}
                \caption{Expected total throughput with respect to $T$ under Rayleigh fading.}
                \label{fig:rayl_T}
        \end{subfigure}
        \hfill
        \begin{subfigure}[h]{0.45\textwidth}
                \includegraphics[width=\textwidth]{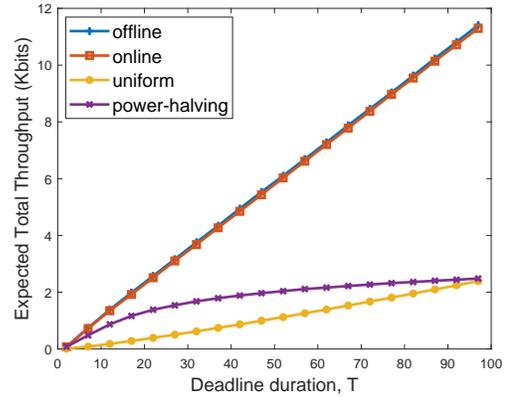}
                \caption{Expected total throughput with respect to $T$ under G-E model. }
                \label{fig:g-e_T}
        \end{subfigure}
				\caption{Expected total throughput of the WPD with $N = 20$ channel levels with respect to the frame length, $T$.}
\label{fig:TH_T}
\end{figure}
Finally, we illustrate the transmission rate of the WPD in units of bits per seconds (bits/sec) in Figure \ref{fig:rate_T}. IT can be seen from both Figure \ref{fig:rate_T_rayl} and Figure \ref{fig:rate_T_ge} that the online policy has a significantly higher rate than the uniform and power-halving policies. It is also evident that on the average, the online policy achieves a significantly good performance with respect to the offline policy.
\begin{figure}
        \centering
        \begin{subfigure}[h]{0.45\textwidth}
                \includegraphics[width=\textwidth]{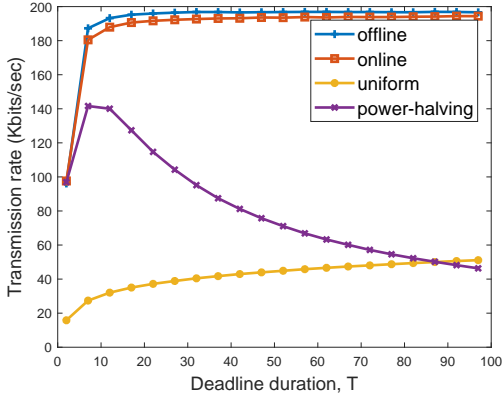}
                \caption{Expected transmission rate of the WPD with respect to $T$ under Rayleigh fading.}
                \label{fig:rate_T_rayl}
        \end{subfigure}
        \hfill
        \begin{subfigure}[h]{0.45\textwidth}
                \includegraphics[width=\textwidth]{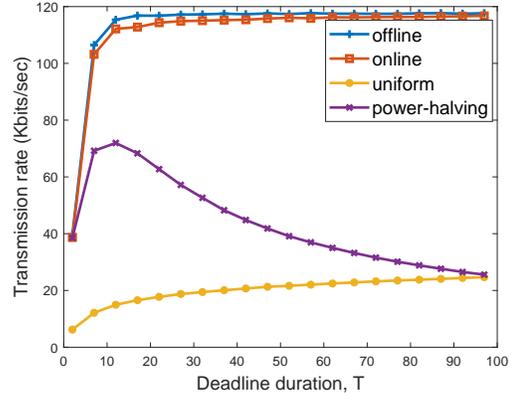}
                \caption{Expected transmission rate of the WPD with respect to $T$ under G-E model. }
                \label{fig:rate_T_ge}
        \end{subfigure}
				\caption{Expected transmission rate of the WPD with $N = 20$ channel levels with respect to the frame length, $T$.}
\label{fig:rate_T}
\end{figure}

\subsection{MAB}
Here, we evaluate the performance of TS and $\epsilon$-greedy algorithms and compare their performance. In Figure \ref{fig:MAB}, we plot the per-period regret of both algorithms. For plots, we use the following synthetic parameters; $T=15$, $N=30$, $L = 1000,\ 2500,\ 3000$ bits, $Z = 500,\ 700,\ 750$, and $\mathcal{E} =1,\ 3,\ 4$ $\mu$Joules. Per-period regret is the gap between the optimal utility and the utility achieved by the given algorithm. We obtain the value of the optimal utility by exhaustive search for comparison purposes only. Each point in Figure \ref{fig:MAB} is averaged over $10^5$ samples.

The greedy algorithm ($\epsilon=0$) has the worst performance as it does not explore at all. By giving non-zero values for $\epsilon$, we can see that $0.05$-greedy and $0.1$-greedy greatly improve upon the greedy algorithm by performing explorations. However, we see a poor performance regarding their convergence rate. TS improves the convergence rate significantly by simply adding intelligence to the exploration phase. This makes the TS algorithm to approach a per-period regret of $0$ considerably faster than the $\epsilon$-greedy algorithm.
\begin{figure}[ht]
  \centering
    \includegraphics[scale=.5]{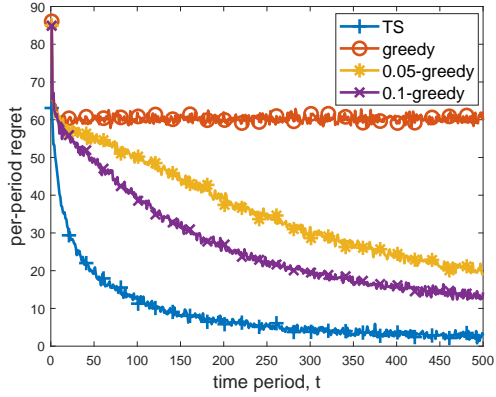}
		  \caption{Per-period regret comparison of TS and $\epsilon$-greedy algorithms for $\epsilon = 0,0.05,0.1$.}
			\label{fig:MAB}
\end{figure}

\section{Conclusions}

In this work, we studied the a WPCN scenario operating in a finite horizon. An AP transmits an RF signal to energize an WPD for a certain duration (WPT period) and then it stops sending energy and collects data from the WPD in the remainder of the horizon (IT period). The wireless channel varies randomly over the horizon and it is only available to the WPD causally and only in the IT period. We first derived an upper-bound for the performance of the network in an offline manner by assuming that the channel realizations are available non-causally at the AP.  We then studied the online counterpart of the problem by assuming that the channel realization are available only causally and in the IT period. We show that there exist a  time-dependent threshold on the energy level of the WPD in which it is optimal to stop WPT and start the IT period. Then, we show that the optimal power allocation in the IT period follows a fractional structure in which the WPD at each time slot allocates a fraction of its energy that depends on the immediate channel state as well as certain measure of future expectations. The numerical results show that the online policy achieves  a performance significantly close to the upper-bound. 
We then extended the model by embedding a MEC unit at the AP capable of performing heavy computational tasks. We formulated a Bayesian inference reinforcement learning problem to address the dependency of the application performance coupled with that of physical layer. We show that the Bayesian inference achieves a convergence rate that is much faster than that of the $\epsilon$-greedy algorithm. In the future, we aim to extend the results to the case of multiple WPDs.
\label{CON}

\begin{appendices}
\section{Proof of Theorem \ref{off_optimal}}
\label{proof:offline}
Consider the following concave optimization of the  throughput at time $T-1$ and $T+1$, given that the amount of available energy at time $T-1$ is $E(T-1)$
\begin{align}
\max_{p(T-1),p(T)} & \sqrt[m]{\frac{g(T-1)p(T-1)}{\lambda}}+\sqrt[m]{\frac{g(T)(E(T-1)-p(T-1))}{\lambda}} \nonumber\\
&p(T-1)\leq E(T-1),\hspace{0.5cm}p(T)\leq E(T-1)-p(T-1).\nonumber
\end{align}
The WPD at the last time slot should utilize all the available energy before the transmission frame expires. Hence, we set $p(T) = E(T-1)-p(T-1)$. The optimization problem becomes
\begin{align}
\max_{p(T-1)} & \sqrt[m]{\frac{g(T-1)p(T-1)}{\lambda}}+\sqrt[m]{\frac{g(T)(E(T-1)-p(T-1))}{\lambda}} \nonumber\\
&0\leq p(T-1)\leq E(T-1).\nonumber
\end{align}
The Lagrangian of the above problem can be written as
\begin{align}
\mathcal{L}(p(T-1),\mu_1,\mu_2) = &\sqrt[m]{\frac{g(T-1)p(T-1)}{\lambda}}+\sqrt[m]{\frac{g(T)(E(T-1)-p(T-1))}{\lambda}}\nonumber\\
&-\mu_1(p(T-1)-E(T-1) + \mu_2p(T-1)\nonumber
\end{align}
The derivative of the Lagrangian is calculated as follows
\begin{align}
\frac{\partial \mathcal{L}(p(T-1),\mu_1,\mu_2)}{\partial p(T-1)} =&\frac{1}{m} \sqrt[m]{\frac{g(T-1)}{\lambda}}p(T-1)^{\frac{1}{m}-1}\nonumber\\
&-\frac{1}{m} \sqrt[m]{\frac{g(T)}{\lambda}}(E(T-1)-p(T-1))^{\frac{1}{m}-1}+(\mu_2-\mu_1)\nonumber
\end{align}
Prior to equating the Lagrangian to zero, we assume that the optimal power allocation satisfies the constraint, i.e., $0\leq p^*(T-1)\leq E(T-1)$, and set $\mu_1=\mu_2=0$. By solving the derivative of the relaxed Lagrangian, we get
\begin{align}
p^*(T-1)=\frac{g(T-1)^{\frac{1}{m-1}}}{g(T-1)^{\frac{1}{m-1}} + g(T)^{\frac{1}{m-1}}}E(T-1)\nonumber
\end{align}
Note that since $0\leq\frac{g(T-1)^{\frac{1}{m-1}}}{g(T-1)^{\frac{1}{m-1}} + g(T)^{\frac{1}{m-1}}}\leq 1$, the constraint is satisfied. Let us calculate the optimum sum throughput at time $T-1$ and $T$:

\begin{align}
r(T-1)+r(T)&=\sqrt[m]{\frac{g(T-1)p^*(T-1)}{\lambda}}+\sqrt[m]{\frac{g(T)(E(T-1)-p^*(T-1))}{\lambda}}\nonumber\\
&=\sqrt[m]{\frac{E(T-1)}{\lambda}}\Bigg[\sqrt[m]{\frac{g(T-1)g(T-1)^{\frac{1}{m-1}}}{g(T-1)^{\frac{1}{m-1}} + g(T)^{\frac{1}{m-1}}}} + \sqrt[m]{\frac{g(T)g(T)^{\frac{1}{m-1}}}{g(T-1)^{\frac{1}{m-1}} + g(T)^{\frac{1}{m-1}}}}\Bigg]\nonumber\\
&=\sqrt[m]{\frac{E(T-1)}{\lambda}}\frac{g(T-1)^{\frac{1}{m-1}} + g(T)^{\frac{1}{m-1}}}{\sqrt[m]{g(T-1)^{\frac{1}{m-1}} + g(T)^{\frac{1}{m-1}}}}\nonumber\\
&=\sqrt[m]{\frac{E(T-1)G(T-1)}{\lambda}},\nonumber
\end{align}
where $G(T-1) =\big[g(T-1)^{\frac{1}{m-1}} + g(T)^{\frac{1}{m-1}}\big]^{m-1}$. To generalize the results, we use induction. Suppose that the above results are true for some time $t+1$. Next consider the optimization of sum throughput from time $t$ to $T$:
\begin{align}
\max_{p(t)} & \sqrt[m]{\frac{g(t)p(t)}{\lambda}}+\sqrt[m]{\frac{(E(t)-p(t))G(T-1)}{\lambda}} \nonumber
\end{align}
Similar to the above analysis, it follows that
\begin{align}
p^*(t)=\frac{g(t)^{\frac{1}{m-1}}}{g(t)^{\frac{1}{m-1}} + G(t+1)^{\frac{1}{m-1}}}E(t)\nonumber\\
\sum_{\tau=t}^{T}r(\tau) = \sqrt[m]{\frac{E(t)G(t)}{\lambda}},\nonumber
\end{align}
where $G(t) =\big[g(t)^{\frac{1}{m-1}} + G(t+1)^{\frac{1}{m-1}}\big]^{m-1}$

\section{Proof of Theorem \ref{optimal}}
\label{proof:optimal}
The proof is by induction. We have shown in (\ref{alfa1}), (\ref{Q1}), and (\ref{V1}), that the case for $k=1$ is true. By assuming the the case for $k-1$ is true, let us calculate the case $k$. The value function is given as
\begin{align}
V_{\alpha}(E(T-k),g(T-k)) =& U_{\alpha}(E(T-k),g(T-k)) + \sum q_i V(E(T-(k-1)),g_i)\label{V_a_T_k}
\end{align}
Note that $E(T-(k-1)) = (1-\alpha(T-k))E(T-k)$ and since the case is true for $k-1$, from (\ref{V_k}), we have 
\begin{align}
&V(E(T-(k-1)),g_i)= \sqrt[m]{\frac{(1-\alpha(T-k))E(T-k)}{\lambda}}\big(g_i^{\frac{1}{m-1}} + Q(T-k+1)^{\frac{m}{m-1}}\big)^\frac{m-1}{m} \label{V_T_k_1}
\end{align}
By substituting (\ref{V_T_k_1}) in (\ref{V_a_T_k}) we get
\begin{align}
V_{\alpha}(E(T-k),g(T-k))& = \sqrt[m]{\frac{g(T-k)\alpha(T-k) E(T-k)}{\lambda}}\nonumber\\
&+\sum q_i \sqrt[m]{\frac{(1-\alpha(T-k))E(T-k)}{\lambda}}\times\big(g_i^{\frac{1}{m-1}} + Q(T-k+1)^{\frac{m}{m-1}}\big)^\frac{m-1}{m}
\end{align}
By differentiating with respect to $\alpha(T-k)$ and equating to zero, we obtain:
\begin{align}
\alpha^*(T-k) = &\frac{g(T-k)^\frac{1}{m-1}}{g(T-k)^\frac{1}{m-1} + Q(T-k)^\frac{m}{m-1}},
\end{align}
where 
\begin{align}
Q(T-k) = \sum q_i\big(g_i^{\frac{1}{m-1}} + Q(T-k+1)^{\frac{m}{m-1}}\big)^\frac{m-1}{m}
\end{align}
Hence, (\ref{alfa_k}) and $(\ref{Q_k})$ hold by induction. For the last part let us calculate $V(E(T-k),g(T-k))$ 
\begin{align}
V(E(T-k),g(T-k)) &= \sqrt[m]{\frac{g(T-k)g(T-k)^\frac{1}{m-1}E(T-k)}{\lambda(g(T-k)^\frac{1}{m-1} + Q(T-k)^\frac{m}{m-1})}}\nonumber\\
&+\sum q_i \sqrt[m]{\frac{Q(T-k)^\frac{m}{m-1}E(T-k)}{\lambda(g(T-k)^\frac{1}{m-1} + Q(T-k)^\frac{m}{m-1})}}\big(g_i^{\frac{1}{m-1}} + Q(T-k+1)^{\frac{m}{m-1}}\big)^\frac{m-1}{m}\nonumber\\
&=\sqrt[m]{\frac{E(T-k)}{\lambda(g(T-k)^\frac{1}{m-1} + Q(T-k)^\frac{m}{m-1})}}(g(T-k)^\frac{1}{m-1}+Q(T-k)^\frac{m}{m-1})\nonumber\\
&=\sqrt[m]{\frac{E(T-k)}{\lambda}}\big(g(T-k)^{\frac{1}{m-1}} + Q(T-k )^{\frac{m}{m-1}}\big)^\frac{m-1}{m}.
\end{align}
Thus, (\ref{V_k}) also holds by induction.
\section{Proof of Theorem \ref{theorem:threshold}}
\label{proof_theorem:threshold}
The proof is by induction. We will show that the result of the theorem is true for $J_{t}(E(t))$ for all $t=1,\ldots,T-1$. The result of the theorem is verified for $t=T-1$ in (\ref{gamma_T_1}). Let us assume that the theorem holds for $t+1$, i.e., if $E(t+1)\geq\gamma(t+1)$, it is optimal to stop the EH process, where $\gamma(t+1)$ is the solution to the following equation,
\begin{align}
\sum q_i \sqrt[m]{1+\frac{e_i}{\gamma(t+1)}} = \frac{Q(t)}{Q(t+1)}
\end{align}
At time slot $t$ we have:
\begin{align}
J_{t}(E(t)) = \max\Bigg(\sqrt[m]{\frac{E(t)}{\lambda}} Q(t-1),\mathbb{E}(J_{t+1}(E(t+1))|E(t)\Bigg)
\end{align}
First, let us assume that $E(t)\geq \gamma(t+1)$. Since $E(t+1)\geq E(t)$, it readily follows that $E(t+1)\geq \gamma(t+1)$. Thus, we have
\begin{align}
\mathbb{E}(J_{t+1}(E(t+1))|E(t))=\sum q_i\sqrt[m]{\frac{E(t)+e_i}{\lambda}} Q(t))\label{J_T_K}
\end{align}
Hence,
\begin{align}
J_{t}(E(t)) = \max\Bigg(\sqrt[m]{\frac{E(t)}{\lambda}} Q(t-1),\sum q_i\sqrt[m]{\frac{E(t)+e_i}{\lambda}} Q(t)\Bigg)
\end{align}
Since, $Q(t-1)>Q(t)$, if $E(t)\geq \gamma(t)$, then it is optimal to stop the EH process, and $\gamma(t)$ is the solution of,
\begin{align}
\sum q_i \sqrt[m]{1+\frac{e_i}{\gamma(t)}} = \frac{Q(t-1)}{Q(t)}.\label{gama_app}
\end{align}
Note that the left hand side of (\ref{gama_app}) is strictly decreasing with respect to $\gamma(t)$ and its range is $(1\  \infty)$. Since $\frac{Q(t-1)}{Q(t)}>1$ is proved in Lemma \ref{lemma:Q_decreasing}, there is a unique solution for $\gamma(t)$ satisfying (\ref{gama_app}).
Thus, if $E(t)\geq\gamma(t+1)$, then the theorem is also true for case $k$. In the following, we will generalize the proof for any value of $E(t)$. Note that if $\gamma(t)>\gamma(t+1)$, then the proof will include any $E(t)$. Because, if $E(t)\geq \gamma(t)$, then,
\begin{align}
E(t+1)\geq E(t)\geq \gamma(t)>\gamma(t+1),
\end{align}
and (\ref{J_T_K}) will hold. Using the results of Lemma \ref{lemma:Qfrac_decreasing} we have
\begin{align}
\sum q_i \sqrt[m]{1+\frac{e_i}{\gamma(t)}} < \sum q_i \sqrt[m]{1+\frac{e_i}{\gamma(t+1)}}
\end{align}
Hence, $\gamma(t)>\gamma(t+1)$, and the theorem holds.

\end{appendices}
\bibliographystyle{IEEEtran}
\bibliography{ref}

\begin{thebibliography}{10}
\providecommand{\url}[1]{#1}
\csname url@samestyle\endcsname
\providecommand{\newblock}{\relax}
\providecommand{\bibinfo}[2]{#2}
\providecommand{\BIBentrySTDinterwordspacing}{\spaceskip=0pt\relax}
\providecommand{\BIBentryALTinterwordstretchfactor}{4}
\providecommand{\BIBentryALTinterwordspacing}{\spaceskip=\fontdimen2\font plus
\BIBentryALTinterwordstretchfactor\fontdimen3\font minus
  \fontdimen4\font\relax}
\providecommand{\BIBforeignlanguage}[2]{{%
\expandafter\ifx\csname l@#1\endcsname\relax
\typeout{** WARNING: IEEEtran.bst: No hyphenation pattern has been}%
\typeout{** loaded for the language `#1'. Using the pattern for}%
\typeout{** the default language instead.}%
\else
\language=\csname l@#1\endcsname
\fi
#2}}
\providecommand{\BIBdecl}{\relax}
\BIBdecl

\bibitem{6951347}
X.~Lu, P.~Wang, D.~Niyato, D.~I. Kim, and Z.~Han, ``Wireless networks with rf
  energy harvesting: A contemporary survey,'' \emph{IEEE Communications Surveys
  Tutorials}, vol.~17, no.~2, pp. 757--789, Secondquarter 2015.

\bibitem{6678102}
H.~Ju and R.~Zhang, ``Throughput maximization in wireless powered communication
  networks,'' \emph{IEEE Transactions on Wireless Communications}, vol.~13,
  no.~1, pp. 418--428, January 2014.

\bibitem{8048675}
X.~Di, K.~Xiong, P.~Fan, H.~Yang, and K.~B. Letaief, ``Optimal resource
  allocation in wireless powered communication networks with user
  cooperation,'' \emph{IEEE Transactions on Wireless Communications}, vol.~16,
  no.~12, pp. 7936--7949, Dec 2017.

\bibitem{7866871}
D.~Xu and Q.~Li, ``Joint power control and time allocation for wireless powered
  underlay cognitive radio networks,'' \emph{IEEE Wireless Communications
  Letters}, vol.~6, no.~3, pp. 294--297, June 2017.

\bibitem{7676282}
A.~Biason and M.~Zorzi, ``Battery-powered devices in wpcns,'' \emph{IEEE
  Transactions on Communications}, vol.~65, no.~1, pp. 216--229, Jan 2017.

\bibitem{7492928}
M.~A. Abd-Elmagid, A.~Biason, T.~ElBatt, K.~G. Seddik, and M.~Zorzi, ``On
  optimal policies in full-duplex wireless powered communication networks,'' in
  \emph{2016 14th International Symposium on Modeling and Optimization in
  Mobile, Ad Hoc, and Wireless Networks (WiOpt)}, May 2016, pp. 1--7.

\bibitem{7996351}
{M. A. Abd-Elmagid and A. Biason and T. ElBatt and K. G. Seddik and M. Zorzi},
  ``Non-orthogonal multiple access schemes in wireless powered communication
  networks,'' in \emph{2017 IEEE International Conference on Communications
  (ICC)}, May 2017, pp. 1--6.

\bibitem{7417596}
M.~A. Abd-Elmagid, T.~ElBatt, and K.~G. Seddik, ``Optimization of wireless
  powered communication networks with heterogeneous nodes,'' in \emph{2015 IEEE
  Global Communications Conference (GLOBECOM)}, Dec 2015, pp. 1--7.

\bibitem{7100855}
C.~Zhong, G.~Zheng, Z.~Zhang, and G.~K. Karagiannidis, ``Optimum wirelessly
  powered relaying,'' \emph{IEEE Signal Processing Letters}, vol.~22, no.~10,
  pp. 1728--1732, Oct 2015.

\bibitem{7018086}
H.~Chen, Y.~Li, J.~L. Rebelatto, B.~F. Uchôa-Filho, and B.~Vucetic,
  ``Harvest-then-cooperate: Wireless-powered cooperative communications,''
  \emph{IEEE Transactions on Signal Processing}, vol.~63, no.~7, pp.
  1700--1711, April 2015.

\bibitem{7248986}
Y.~Gu, H.~Chen, Y.~Li, and B.~Vucetic, ``An adaptive transmission protocol for
  wireless-powered cooperative communications,'' in \emph{2015 IEEE
  International Conference on Communications (ICC)}, June 2015, pp. 4223--4228.

\bibitem{7037009}
H.~Ju and R.~Zhang, ``User cooperation in wireless powered communication
  networks,'' in \emph{2014 IEEE Global Communications Conference}, Dec 2014,
  pp. 1430--1435.

\bibitem{7500446}
M.~Zhong, S.~Bi, and X.~Lin, ``User cooperation for enhanced throughput
  fairness in wireless powered communication networks,'' in \emph{2016 23rd
  International Conference on Telecommunications (ICT)}, May 2016, pp. 1--6.

\bibitem{1708-08810}
\BIBentryALTinterwordspacing
S.~Bi and Y.~J. Zhang, ``Computation rate maximization for wireless powered
  mobile-edge computing with binary computation offloading,'' \emph{CoRR}, vol.
  abs/1708.08810, 2017. [Online]. Available:
  \url{http://arxiv.org/abs/1708.08810}
\BIBentrySTDinterwordspacing

\bibitem{7417552}
C.~You and K.~Huang, ``Wirelessly powered mobile computation offloading: Energy
  savings maximization,'' in \emph{2015 IEEE Global Communications Conference
  (GLOBECOM)}, Dec 2015, pp. 1--6.

\bibitem{8234686}
F.~Wang, J.~Xu, X.~Wang, and S.~Cui, ``Joint offloading and computing
  optimization in wireless powered mobile-edge computing systems,'' \emph{IEEE
  Transactions on Wireless Communications}, vol.~PP, no.~99, pp. 1--1, 2017.

\bibitem{7997477}
{F. Wang and J. Xu and X. Wang and S. Cui}, ``Joint offloading and computing
  optimization in wireless powered mobile-edge computing systems,'' in
  \emph{2017 IEEE International Conference on Communications (ICC)}, May 2017,
  pp. 1--6.

\bibitem{5992841}
O.~Ozel, K.~Tutuncuoglu, J.~Yang, S.~Ulukus, and A.~Yener, ``Transmission with
  energy harvesting nodes in fading wireless channels: Optimal policies,''
  \emph{IEEE Journal on Selected Areas in Communications}, vol.~29, no.~8, pp.
  1732--1743, September 2011.

\bibitem{6094139}
J.~Yang and S.~Ulukus, ``Optimal packet scheduling in an energy harvesting
  communication system,'' \emph{IEEE Transactions on Communications}, vol.~60,
  no.~1, pp. 220--230, January 2012.

\bibitem{6897968}
Z.~Wang, V.~Aggarwal, and X.~Wang, ``Power allocation for energy harvesting
  transmitter with causal information,'' \emph{IEEE Transactions on
  Communications}, vol.~62, no.~11, pp. 4080--4093, Nov 2014.

\bibitem{7008488}
M.~L. Ku, Y.~Chen, and K.~J.~R. Liu, ``Data-driven stochastic models and
  policies for energy harvesting sensor communications,'' \emph{IEEE Journal on
  Selected Areas in Communications}, vol.~33, no.~8, pp. 1505--1520, Aug 2015.

\bibitem{7032105}
R.~Ma and W.~Zhang, ``Optimal power allocation for energy harvesting
  communications with limited channel feedback,'' in \emph{2014 IEEE Global
  Conference on Signal and Information Processing (GlobalSIP)}, Dec 2014, pp.
  193--197.

\bibitem{7865904}
M.~R. Zenaidi, Z.~Rezki, and M.~S. Alouini, ``Performance limits of online
  energy harvesting communications with noisy channel state information at the
  transmitter,'' \emph{IEEE Access}, vol.~5, pp. 1239--1249, 2017.

\bibitem{7736112}
W.~Du, J.~C. Liando, H.~Zhang, and M.~Li, ``Pando: Fountain-enabled fast data
  dissemination with constructive interference,'' \emph{IEEE/ACM Transactions
  on Networking}, vol.~25, no.~2, pp. 820--833, April 2017.

\bibitem{7959595}
B.~T. Bacinoglu, E.~Uysal-Biyikoglu, and C.~E. Koksal, ``Finite-horizon
  energy-efficient scheduling with energy harvesting transmitters over fading
  channels,'' \emph{IEEE Transactions on Wireless Communications}, vol.~16,
  no.~9, pp. 6105--6118, Sept 2017.

\bibitem{sutton}
\BIBentryALTinterwordspacing
R.~S. Sutton and A.~G. Barto, \emph{Reinforcement Learning: An
  Introduction}.\hskip 1em plus 0.5em minus 0.4em\relax Cambridge, MA, USA: MIT
  Press, 1998. [Online]. Available:
  \url{http://www.cs.ualberta.ca/%7Esutton/book/ebook/the-book.html}
\BIBentrySTDinterwordspacing

\bibitem{6574874}
W.~Zhang, Y.~Wen, K.~Guan, D.~Kilper, H.~Luo, and D.~O. Wu, ``Energy-optimal
  mobile cloud computing under stochastic wireless channel,'' \emph{IEEE
  Transactions on Wireless Communications}, vol.~12, no.~9, pp. 4569--4581,
  September 2013.

\bibitem{7880703}
S.~Ko, K.~Huang, S.~Kim, and H.~Chae, ``Live prefetching for mobile computation
  offloading,'' \emph{IEEE Transactions on Wireless Communications}, vol.~16,
  no.~5, pp. 3057--3071, May 2017.

\bibitem{4674675}
J.~Lee and N.~Jindal, ``Energy-efficient scheduling of delay constrained
  traffic over fading channels,'' \emph{IEEE Transactions on Wireless
  Communications}, vol.~8, no.~4, pp. 1866--1875, April 2009.

\bibitem{4357594}
M.~Zafer and E.~Modiano, ``Delay-constrained energy efficient data transmission
  over a wireless fading channel,'' in \emph{2007 Information Theory and
  Applications Workshop}, Jan 2007, pp. 289--298.

\bibitem{8403495}
C.~You, Y.~Zeng, R.~Zhang, and K.~Huang, ``Resource management for asynchronous
  mobile-edge computation offloading,'' in \emph{2018 IEEE International
  Conference on Communications Workshops (ICC Workshops)}, May 2018, pp. 1--6.

\bibitem{1801-03668}
\BIBentryALTinterwordspacing
{Changsheng You and Yong Zeng and Rui Zhang and Kaibin Huang}, ``Asynchronous
  mobile-edge computation offloading: Energy-efficient resource management,''
  \emph{CoRR}, vol. abs/1801.03668, 2018. [Online]. Available:
  \url{http://arxiv.org/abs/1801.03668}
\BIBentrySTDinterwordspacing

\bibitem{wearable}
M.~R.~N. Edward~Sazonov, \emph{Wearable Sensors: Fundamentals, Implementation
  and Applications}, 1st~ed.\hskip 1em plus 0.5em minus 0.4em\relax Academic
  Press, 2014.

\bibitem{ishadow}
\BIBentryALTinterwordspacing
A.~Mayberry, P.~Hu, B.~Marlin, C.~Salthouse, and D.~Ganesan, ``ishadow: Design
  of a wearable, real-time mobile gaze tracker,'' in \emph{Proceedings of the
  12th Annual International Conference on Mobile Systems, Applications, and
  Services}, ser. MobiSys '14.\hskip 1em plus 0.5em minus 0.4em\relax New York,
  NY, USA: ACM, 2014, pp. 82--94. [Online]. Available:
  \url{http://doi.acm.org/10.1145/2594368.2594388}
\BIBentrySTDinterwordspacing

\bibitem{Rostami}
\BIBentryALTinterwordspacing
M.~Rostami, J.~Gummeson, A.~Kiaghadi, and D.~Ganesan, ``Polymorphic radios: A
  new design paradigm for ultra-low power communication,'' in \emph{Proceedings
  of the 2018 Conference of the ACM Special Interest Group on Data
  Communication}, ser. SIGCOMM '18.\hskip 1em plus 0.5em minus 0.4em\relax New
  York, NY, USA: ACM, 2018, pp. 446--460. [Online]. Available:
  \url{http://doi.acm.org/10.1145/3230543.3230571}
\BIBentrySTDinterwordspacing

\bibitem{arxiv:TS}
\BIBentryALTinterwordspacing
D.~Russo, B.~V. Roy, A.~Kazerouni, and I.~Osband, ``A tutorial on thompson
  sampling,'' \emph{CoRR}, vol. abs/1707.02038, 2017. [Online]. Available:
  \url{http://arxiv.org/abs/1707.02038}
\BIBentrySTDinterwordspacing

\bibitem{MCMC}
G.~C. Christian~Robert, \emph{Monte Carlo Statistical Methods}, 2nd~ed., ser.
  Springer Texts in Statistics.\hskip 1em plus 0.5em minus 0.4em\relax
  Springer, 2004.

\end{thebibliography}

\end{document}